\documentclass[runningheads]{llncs}
\usepackage[T1]{fontenc}
\usepackage{graphicx}
\usepackage[a4paper, left=25mm, right=25mm, top=25mm, bottom=25mm]{geometry}
\usepackage[style=ext-numeric,citestyle=numeric-comp,sorting=nyt,sortcites, backend=biber, giveninits=true, maxnames=99]{biblatex}
\DeclareFieldFormat[article,periodical]{volume}{\mkbibbold{#1}}
\DeclareFieldFormat[article,periodical]{number}{\mkbibparens{#1}}
\renewbibmacro{in:}{%
  \ifentrytype{article}
    {}
    {\printtext{\bibstring{in}\intitlepunct}}}

\DeclareFieldFormat{issuedate}{#1}
\renewcommand{\jourvoldelim}{\addcomma\space}

\renewbibmacro*{journal+issuetitle}{%
  \usebibmacro{journal}%
  \setunit*{\jourvoldelim}%
  \iffieldundef{series}
    {}
    {\setunit*{\jourserdelim}%
     \printfield{series}%
     \setunit{\servoldelim}}%
  \usebibmacro{volume+number+eid}%
  \setunit{\addcolon\space}%
  \usebibmacro{issue}%
  \setunit{\bibpagespunct}%
  \printfield{pages}%
  \setunit{\space}%
  \usebibmacro{issue+date}%
  \newunit}

\renewbibmacro*{note+pages}{%
  \printfield{note}%
  \newunit}

\DeclareFieldFormat[article,periodical]{date}{\mkbibparens{#1}}

\renewbibmacro*{publisher+location+date}{%
  \printlist{publisher}%
  \iflistundef{location}
    {\setunit*{\addcomma\space}}
    {\setunit*{\addcolon\space}}%
  \printlist{location}%
  \setunit*{\addcomma\space}%
  \usebibmacro{date}%
  \newunit}

\AtEveryBibitem{\clearfield{month}}
\AtEveryBibitem{\clearfield{day}}

\AtEveryBibitem{
  \ifentrytype{online}{}{
    \ifentrytype{thesis}{}{
      \ifentrytype{unpublished}{}{ 
        \clearfield{url}
      }
    }
  }
  \clearfield{language}
}

\DeclareSourcemap{
  \maps[datatype=bibtex]{
    \map[overwrite=true]{
      \step[fieldset=language, null]
      \step[fieldset=address, null]
     \step[fieldset=pagetotal, null]
    }
  }
}

\DeclareSourcemap{
    \maps[datatype=bibtex]{
      \map{
        \step[fieldsource=eprint,final]
        \step[fieldset=url,null]
        \step[fieldset=doi,null]
      }  
    }
  }

\AtEveryBibitem{\clearfield{pubstate}} 
\AtEveryBibitem{\clearfield{eprintclass}} 
\AtEveryBibitem{\clearfield{version}} 

\usepackage[unicode=true, pdfusetitle, colorlinks=true, allcolors=blue, bookmarks=false]{hyperref}
\usepackage[hyphenbreaks]{breakurl}
\usepackage{xurl}
\hypersetup{breaklinks=true}
\usepackage{color}

\usepackage{amssymb,amsfonts,stmaryrd,mathtools,mathrsfs}
\usepackage{tikz}
\usetikzlibrary{cd}
\usepackage{lipsum}

\newcommand*{\RR}{\mathbb{R}}

\newcommand\restr[2]{{
  \left.\kern-\nulldelimiterspace 
  #1 
  \right|_{#2} 
}}

\newcommand{\T}{{\mathsf T}}
\newcommand{\cT}{\T^{\ast}}
\newcommand*{\dd}{\mathrm{d}}

\newcommand*{\liedv}[1]{\mathcal{L}_{#1}}

\newcommand{\Cinfty}{\mathscr{C}^\infty}

\DeclareMathOperator{\Ad}{Ad}

\mathtoolsset{showonlyrefs} 

\usepackage[colorinlistoftodos]{todonotes} 

\begin{document}
\title{Reduction of hybrid Hamiltonian systems\\ with non-equivariant momentum maps}
\titlerunning{Reduction of hybrid Hamiltonian systems with non-equivariant momentum maps}
%

\author{Leonardo Colombo\inst{1} \and
María Emma Eyrea Iraz\'u\inst{2} \and
María Eugenia Garc\'ia\inst{2} \and\\ Asier López-Gordón\inst{3}\and
Marcela Zuccalli\inst{2}}
\authorrunning{L.~Colombo, M.~E.~Eyrea Irazú, A.~López-Gordón, and M.~E.~Zuccalli}
%
\institute{Centre for  Automation and Robotics, Spanish Research Council, Arganda del Rey, Spain.
\email{leonardo.colombo@csic.es}\\
\and
Departamento de Matem\'atica, Centro de Matemática (CMaLP), Facultad de Ciencias Exactas, Universidad Nacional de La Plata, Argentina.\\
\email{maemma@mate.unlp.edu.ar, maru@mate.unlp.edu.ar, marce@mate.unlp.edu.ar}\and
Institute of Mathematics, Polish Academy of Sciences, Warsaw, Poland. \\
\email{alopez-gordon@impan.pl}
}
\maketitle              
\begin{abstract}
We develop a reduction scheme \textit{à la} Marsden--Weinstein--Meyer for hybrid Hamiltonian systems. Our method does not require the momentum map to be equivariant, neither to be preserved by the impact map. We illustrate the applicability of our theory with an example.

\end{abstract}
\section{Introduction}

Hybrid Hamiltonian systems are dynamical systems that combine continuous-time Hamiltonian dynamics with discrete transitions or impacts. Ames and Sastry \cite{ames_hybrid_2006} developed a reduction method \textit{à la} Marsden--Weinstein--Meyer for hybrid Hamiltonian systems with a Lie group of symmetries. However, they required the level sets of the momentum map to be preserved by the impact map. In a previous work \cite{colombo_generalized_2022}, we replaced this restrictive assumption by introducing the notion of generalized hybrid momentum map --which, roughly speaking, means that the impact map takes level sets of the momentum map into other level sets of the momentum map. In this paper, we go one step further by removing the requirement of the momentum map to be equivariant with respect to the co-adjoint action. The crucial idea is that, given a momentum map which is not equivariant with respect to the coadjoint action of a Lie group on the dual of it Lie algebra, there is a modified action, the so-called affine action with respect to which the momentum map is equivariant.

The remainder of the paper is structured as follows. Section \ref{sec2} introduces hybrid mechanical systems. In Section \ref{sec3} we introduce Hamiltonian hybrid $G$-spaces and the notion of the generalized hybrid momentum map, in order to perform the reduction by symmetries in Section \ref{sec4} for hybrid Hamiltonian systems with non-equivariant momentum map. We illustrate the reduction procedure in Section \ref{sec5} with an example. 

\section{Hybrid mechanical systems}\label{sec2}

A simple hybrid systems is a $4$-tuple $\mathcal{H}=(D, X, S, \Delta)$, where $D$ is a smooth manifold (called the \textit{domain}), $X$ is a smooth vector field on $D$, $S$ is an embedded submanifold of $D$ with co-dimension $1$ (called the \textit{switching surface} or the \textit{guard}), and $\Delta:S\to D$ is a smooth embedding (called the \textit{impact map} or the \textit{reset map}). The triple $(D,S,\Delta)$ is called \textit{hybrid manifold}.

The dynamical system generated by $\mathcal{H}$ is
\begin{equation*}\label{LHS}\Sigma_{\mathcal{H}}:\begin{cases} \dot{\gamma}(t)=X(\gamma(t))\, ,\quad\quad \gamma^{-}(t)\notin{S}\, , \\ \gamma^{+}(t)=\Delta(\gamma^{-}(t))\, ,\quad \gamma^-(t)\in{S}\, . \end{cases}\end{equation*} where $\gamma:I\subset\mathbb{R}\to D$, and $\gamma^{-}(t)\coloneqq \displaystyle{\lim_{\tau\to t^{-}}}x(\tau)$,\, $\gamma^{+}(t)\coloneqq \displaystyle{\lim_{\tau\to t^{+}}}x(\tau)$.

A solution of a simple hybrid system may experience a Zeno state if infinitely many impacts occur in a finite amount of time. In order to avoid that, we will hereafter assume that $\overline{\Delta({S})}\cap{S}=\emptyset$, where $\overline{\Delta({S})}$ denotes the closure of $\Delta({S})$, and that the set of impact times is closed and discrete. Refer to \cite{goodman_existence_2020} for details.

 The \textit{hybrid flow} $\chi^{\mathcal{H}}=(\Lambda,\mathcal{J},\mathcal{C})$ of $\mathcal{H}$ consists of:
\begin{enumerate} 
    \item A discrete indexing set $\Lambda=\{0,1,2,\ldots\}\subseteq \mathbb{N}$.
    \item A set of intervals $\mathcal{J}=\{I_{i}\}_{i\in \Lambda}$, called hybrid intervals, where $I_{i}=[\tau_{i},\tau_{i+1}]$ if $i, i+1\in \Lambda$; and $I_{N-1}=[\tau_{N-1},\tau_{N}]$ or $[\tau_{N-1},\tau_{N})$ or $[\tau_{N-1},\infty)$ if $|\Lambda|=N<\infty$, with $\tau_{i},\tau_{i+1},\tau_{N}\in \mathbb{R}$ and $\tau_{i}\leq \tau_{i+1}$,
    \item A collection of solutions $\mathcal{C}=\{c_{i}\}_{i\in \Lambda}$ for the vector field $X$ specifying the continuous-time dynamics, i.e., $\dot{c_{i}}=X(c_{i}(t))$ for all $i\in \Lambda$,
    and such that for each $i,i+1\in \Lambda$, (i) $c_{i}(\tau_{i+1})\in S$, and (ii) $\Delta(c_{i}(\tau_{i+1}))=c_{i+1}(\tau_{i+1})$.
\end{enumerate}

Let us recall that a Hamiltonian system $(D,\omega,H)$ is a symplectic manifold $(D,\omega)$ equipped with a Hamiltonian function $H\in \Cinfty(D)$. Its dynamics is given by the Hamiltonian vector field $X_H$ of $H$, determined by $\omega(X_H, \cdot) = \dd H$. Around each point $x\in D$, there is a system of local Darboux (or canonical) coordinates $(q^i, p_i)$, in which $\omega = \dd q^i \wedge \dd p_i$ and 
$\displaystyle{X_H=\frac{\partial H}{\partial p_i}\frac{\partial}{\partial q^i}-\frac{\partial H}{\partial q^i}\frac{\partial}{\partial p_i}}$. The archetype of a symplectic manifold is the cotangent bundle $\cT Q$ of a manifold $Q$ endowed with the symplectic form $\omega_Q = - \dd \theta_Q$, with $\theta_Q$ the canonical one-form. In bundle coordinates $(q^i, p_i)$, we have that $\theta_Q=p_i \dd q^i$, so they are Darboux coordinates for $\omega_Q$.

A simple hybrid system $\mathcal{H}=(D, X, {S}, \Delta)$ is said to be a \textit{simple hybrid Hamiltonian system} if $X=X_H$ is the Hamiltonian vector field associated with a Hamiltonian system $(D,\omega,H)$.

\section{Hybrid Hamiltonian $G$-spaces}
\label{sec3}

Let $G$ be a finite-dimensional Lie group, with identity element $e\in G$. 
Let $\mathfrak{g}$ be the Lie algebra of $G$, with dual $\mathfrak{g}^*$. Given a (left) Lie group action $\Phi\colon G\times D\to D$, for each $\xi \in \mathfrak{g}$, its associated infinitesimal generator on $D$ is the vector field $\xi_D\in \mathfrak{X}(D)$ given by $\xi_{D}(x)=\restr{\frac{\dd}{\dd t}}{t=0}\Phi(\exp(t\xi),x)$ for each $x\in D$. The adjoint action of $G$ on $\mathfrak{g}$ is given by $\Ad_g \xi = \restr{\frac{\dd}{\dd t}}{t=0} g \exp(t\xi) g^{-1}$ for each $g\in G$ and each $\xi\in \mathfrak{g}$. The co-adjoint action $\Ad_{g^{-1}}^\ast \mu$ of $g\in G$ on $\mu\in \mathfrak{g}^\ast$ is determined by $\langle \Ad^{*}_{g^{-1}}\mu,\xi\rangle=\langle\mu,\Ad_{g^{-1}}\xi\rangle$ for any $\xi \in \mathfrak{g}$.

Hereinafter, we will assume that all the actions considered are free and proper.

A (left) Lie group action $\Phi\colon G\times D\to D$ on a symplectic manifold $(D,\omega)$ is called symplectic if $\Phi_g$ is a symplectomorphism for each $g\in G$, i.e., $\Phi_g^\ast \omega = \omega$. If $\Phi$ is a symplectic action on $(D, \omega)$, a momentum map is a map $J\colon D \to \mathfrak{g}^\ast$ satisfying $\omega(\xi_D, \cdot) = \dd \big(\langle J, \xi \rangle\big)$ for all $\xi\in \mathfrak{g}$.
The $4$-tuple $(D,\omega, \Phi, J)$ is called a Hamiltonian $G$-space. This notion was introduced by Ames and Sastry in \cite{ames_hybrid_2006}. For each $g\in G$ and $\xi\in\mathfrak{g}$, the function $\Psi_{g,\xi}: D\rightarrow \mathbb{R}$ is defined by $$\Psi_{g,\xi}(x)=J_{\xi}(\Phi_{g}(x))-\Ad^{*}_{g^{-1}}(J_{\xi}(x))\, ,$$
where $J_\xi(x) = \langle J(x), \xi\rangle$. One can show that $\Psi_{g,\xi}$ is constant on $D$ (see, for instance, \cite{ortega_momentum_2004}). Hence, the map $\sigma:G\rightarrow \mathfrak{g}^{*}$ determined by $\langle\sigma(g),\xi\rangle = \Psi_{g,\xi}(x)$, for any $\xi \in \mathfrak{g}$ and any $x\in D$, is well-defined. The map $\sigma$ is called the \textit{co-adjoint cocycle} associated with $J$. Moreover, it can be shown that 
$$\Psi:(G, \mathfrak{g}^*)\ni(g,\mu)\mapsto \Ad_{g^{-1}}^{*} \mu+\sigma(g)\in\mathfrak{g}^*$$
is an action of $G$ on $\mathfrak{g}^{*}$, the so-called affine action, and $J$ is equivariant with respect to this action.

Consider the simple hybrid system $\mathcal{H}=(D,X,S,\Delta)$. A Lie group action $\Phi:G\times D\rightarrow D$ of $G$ on $D,$ is called a \textit{hybrid action} if   
$\restr{\Phi}{G\times S}$ is a Lie group action of $G$ on $S$ and the impact map is equivariant with respect to this action, namely, $\Delta\circ \restr{\Phi_g}{S}=\Phi_g\circ \Delta$ for all $g\in G$. Suppose that $\mathcal{H}=(D,X_H,S,\Delta)$ is a hybrid Hamiltonian system with associated Hamiltonian system $(D, \omega, H)$, and assume that the action $\Phi$ is hybrid and symplectic. A momentum map $J$ is called a \emph{generalized hybrid momentum map} for $\mathcal{H}$ if, for each regular value $\mu_-$ of $J$, and each connected component $C$ of $S$,
\begin{equation}
  \Delta \left(\restr{J}{C}^{-1}(\mu_-)  \right) \subset J^{-1}(\mu_+),
  \label{generalized_hybrid_momentum}
\end{equation}
for some regular value $\mu_+$. In particular, when $\mu_+=\mu_-$ for each $\mu^-$ (i.e., $\Delta$ preserves the momentum map), $J$ is called a hybrid momentum map. See \cite{colombo_generalized_2022} for more details. The tuple $(D,S,\Delta,\omega,\Phi,J)$ will be called a hybrid Hamiltonian $G$-space. We will call $\mu\in \mathfrak{g}^\ast$ a \emph{hybrid regular value} if it is a regular value of both $J$ and $\restr{J}{S}$.

\section{Non-equivariant reduction for hybrid Hamiltonian systems }\label{sec4}

In \cite{ames_hybrid_2006}, the Marsden--Weinstein--Meyer reduction of hybrid Hamiltonian systems was studied, while in \cite{marsden_hamiltonian_2007} the Marsden--Weinstein--Meyer reduction for non-equivariant momentum maps was developed. We use the latter result for continuous-time Hamiltonian systems to prove the existence of a reduced hybrid Hamiltonian system given a hybrid Hamiltonian system together with hybrid Hamiltonian $G$-space. Moreover, we are able to prove a relationship between the hybrid flows of these two systems.

For the hybrid manifold $D^{\mathcal{H}}=(D,{S},\Delta),$ a Lie group $G$ and $\Phi$ a hybrid action, the hybrid orbit space is 		$D^{\mathcal{H}}/G=(D/G,{S}/G,\hat{\Delta})$ where $D/G$ and $S/G$ are the orbit space of $\Phi$ and  $\Phi| _{ {S}}$, respectively, and $\hat{\Delta}: {S}/G\rightarrow D/G$ is the induced impact map. If $\Phi:G\times D\rightarrow D$ is a free and proper action and also a hybrid action, then $D^{\mathcal{H}}/G$ is a hybrid manifold and there exist a submersion  $\pi:D\rightarrow D/G$ such that the following diagram commutes:	
    \[
    \begin{tikzcd}\label{variedadesreducidas}
    D\arrow[d, "\pi"] &  {S}\arrow[d, "\pi|_{S}"]\arrow[r, "\Delta"]\arrow[l,swap,hook',"i"]& D\arrow[d, "\pi"]& \\
    D/G &  {S}/G\arrow[l,swap,hook',"i_{G}"]\arrow[r, "\hat{\Delta}"] & D/G&
    \end{tikzcd} 
    \]

Next, we refine the hybrid Marsden--Weinstein--Meyer reduction theory of \cite{ames_hybrid_2006} by relaxing certain technical conditions on hybrid momentum maps. To achieve this, we develop a framework for affine Lie group actions on hybrid momentum maps, eliminating the need for their co-adjoint equivariance. Additionally, we introduce generalized hybrid momentum maps, which extend those in \cite{ames_hybrid_2006} by allowing more flexibility beyond simply requiring that the momentum remains the same before and after impact.

\emph{Non-equivariant reduction:} We now describe how to perform reduction for a momentum map that may have a nonzero non-equivariance one-cocycle. This result is essential for the Hamiltonian reduction by stages construction (see \cite{marsden_hamiltonian_2007}). Let \( J:D\rightarrow \mathfrak{g}^* \) be a momentum map associated with the action \( \Phi \) on a connected symplectic manifold \( D \), with a non-equivariance group one-cocycle \( \sigma \). Consider the affine action \( \Psi \) on \( \mathfrak{g}^\ast \) defined by  
\begin{equation}\label{eq:modified_coadjoint_action}
    \Psi:(g,\mu)\mapsto \Ad_{g^{-1}}^{*} \mu+\sigma(g)\, .
\end{equation}
With respect to this modified action, the momentum map \( J \) becomes equivariant.  Let \( \tilde{G}_{\mu} \) denote the isotropy subgroup of \( \mu \in \mathfrak{g}^* \) under the action \( \Psi \), given by
$$\tilde{G}_{\mu}=\{g\in G: \Psi(g,\mu)=\Ad^{*}_{g^{-1}}\mu+\sigma(g)=\mu\}\, .$$
Under standard regularity assumptions--namely, that \( G \) acts freely and properly on \( D \), that \( \mu \) is a regular value of \( J \), and that \( \tilde{G}_{\mu} \) acts freely and properly on \( J^{-1}(\mu) \)--the quotient manifold \(D_{\mu} = J^{-1}(\mu)/\tilde{G}_{\mu}\) is a symplectic manifold. Its symplectic form \( \omega_{\mu} \) is uniquely determined by the condition \(i^{*}_{\mu}(\omega) = \pi^{*}_{\mu}(\omega_{\mu}),\)  
where \( i_{\mu}:J^{-1}(\mu)\rightarrow D \) is the natural inclusion and \( \pi_{\mu}:J^{-1}(\mu)\to J^{-1}(\mu)/\tilde{G}_{\mu} \) is the canonical projection. Summarizing, if $(D,\omega,\Phi,J)$ is a Hamiltonian $G$-space and the regularity assumptions above hold, then $(D_\mu, \omega_\mu)$ is a symplectic manifold, the so-called reduced phase space.

\begin{proposition}\label{Proposition:isotropy_subgroups}
    Let $(D^{\mathcal{H}},\omega,\Phi,J)$ be a hybrid Hamiltonian $G$-space. Assume that $G$ is connected, and let $\Psi$ denote the affine action \eqref{eq:modified_coadjoint_action}.
    If $\Delta$ is equivariant with respect to $\Phi$, and $\mu_-,\ \mu_+$ are regular values of $J$ such that 
    $$\Delta\left(\restr{J}{S}^{-1}(\mu_-)  \right)\subset J^{-1}(\mu_+)\, ,$$
    then the isotropy subgroups at $\mu_-$ and at $\mu_+$ under the action \( \Psi \) coincide, i.e., $\tilde{G}_{\mu_-}=\tilde{G}_{\mu_+}$.
\end{proposition}


\begin{proof}  Let $g\in \tilde G_{\mu_-}$. Then,
\begin{equation}
\begin{aligned}
  \mu_+
  &=J \circ \Delta \left(\restr{J}{S}^{-1}(\mu_-)  \right)
  = J \circ \Delta\circ \Phi_g \left(\restr{J}{S}^{-1}(\mu_-)  \right)
  \\&
  = J \circ \Phi_g \circ \Delta_H \left(\restr{J}{S_H}^{-1}(\mu_-)  \right)
  \\&
  = \Psi_g \circ J \circ \Delta_H \left(\restr{J}{S_H}^{-1}(\mu_-)  \right)
   = \Psi_g (\mu_+)
  \, ,
\end{aligned}
\end{equation}
where we have used the equivariance of $J$ and $\Delta_H$, so $g\in G_{\mu_+}$, and hence $G_{\mu_-}$ is a Lie subgroup of $G_{\mu_+}$. 

Now, observe that $G_\mu$ has the same dimension, for each $\mu\in \mathfrak g^*$. Therefore, the identity components of $G_{\mu_-}$ and $G_{\mu_+}$ coincide. If we assume that $G$ is connected, $G_{\mu_-}$ and $G_{\mu_+}$ are equal to their identity components, so $G_{\mu_-}=G_{\mu_+}$. 
\flushright$\square$
\end{proof}

\begin{theorem}\label{TEO}
    Let $(D^{\mathcal{H}},\omega,\Phi,J)$ be a hybrid Hamiltonian $G$-space. Assume that \( G \) is connected, and consider a discrete sequence $\Lambda = \left\{\mu_i  \right\}$ of regular values of $J$ such that $\Delta \left(\restr{J}{S}^{-1}(\mu_i)  \right)\subset J^{-1}(\mu_{i+1})$. Let $\tilde{G}_{\mu_i}=\tilde{G}_{\mu_0}$ be the isotropy subgroup in $\mu_i$ (for any $\mu_i$ in the sequence) under the affine action \eqref{eq:modified_coadjoint_action}. Assume that $\Phi$ and $\restr{\Phi}{\tilde{G}_{\mu}\times J^{-1}(\mu)}$ are free and proper actions. Then, for any $\mu_i \in \Lambda$,
    $$(D_{\mu_i}, S_{\mu_i}, \Delta_{\mu_i})\coloneqq \left(J^{-1}(\mu_i)/\tilde{G}_{\mu_i}, \restr{J}{S}^{-1}(\mu_i)/\tilde{G}_{\mu_i}, \restr{\Delta}{\restr{J}{S}^{-1}(\mu_i)}\right)$$
    is a hybrid manifold. 

    Furthermore, if $\mathcal{H}=(D, X, {S}, \Delta)$ is a hybrid Hamiltonian system, and the Hamiltonian $H$ is $\Phi$-invariant, then there is a reduced Hamiltonian $H_{\mu}$ and a reduced hybrid Hamiltonian system   
    $${\mathcal{H}}_{\mu}=\left(\tilde{D}_{\mu},\tilde{S}_{\mu},\tilde{\Delta}_{\mu},\tilde{X}_{H_\mu}\right)\, ,$$
    where the reduced Hamiltonian $H_{\mu}$ on $\tilde{D}_{\mu}$ is uniquely determined by $H_{\mu}\circ \pi_{\mu}=H\circ i_{\mu}$.

    The reduction scheme is summarized in the following commutative diagram:
    \begin{center}
    \begin{tikzcd}
        \cdots \arrow[r] & J^{-1}(\mu_i) \arrow[dd]        &  & \restr{J}{S}^{-1}(\mu_i) \arrow[dd] \arrow[rr, "\restr{\Delta}{J^{-1}(\mu_i)}"] \arrow[ll, hook] &  & J^{-1}(\mu_{i+1}) \arrow[dd]        & \cdots \arrow[l, hook] \\
                         & {} \arrow[d]                      &  &                                                                                               &  &                                       &                        \\
        \cdots \arrow[r] & \frac{J^{-1}(\mu_i)}{\tilde{G}_{\mu_0}} &  & S_{\mu_i} \arrow[rr, "\left(\Delta\right)_{\mu_i}"] \arrow[ll, hook]         &  & \frac{J^{-1}(\mu_{i+1})}{\tilde{G}_{\mu_0}} & \cdots \arrow[l, hook]
    \end{tikzcd}
    \end{center}
\end{theorem}

\begin{proof}
    The crucial idea is that, by Proposition~\ref{Proposition:isotropy_subgroups}, all the regular level sets of $J$ and $\restr{J}{S}$ can be quotiented by the same isotropy subgroup $\tilde{G}_{\mu_0}$. The quotient
    $S_{\mu_i}=\restr{J}{S}^{-1}(\mu_i)/\tilde{G}_{\mu_0}$ is a smooth manifold, since the $\tilde{G}_{\mu_0}$-action restricts to a free and proper action on $S$. As a matter of fact, it is a submanifold of $J^{-1}(\mu_i)/\tilde{G}_{\mu_0}$. Due to its equivariance, $\Delta$ induces an embedding $\Delta_{\mu_i}:S_{\mu_i} \rightarrow J^{-1}(\mu_{i+1})/\tilde{G}_{\mu_0}$.

    \flushright$\square$
\end{proof}

\begin{corollary}
    With ${\mathcal{H}}$ and ${\mathcal{H}}_{\mu}$ as in the previous theorem, if $\chi^{\mathcal{H}}(x_{0})$ is a hybrid flow of ${\mathcal{H}}$ with $x_{0}\in J^{-1}(\mu),$ then there is a corresponding hybrid flow  $\chi^{\mathcal{H}_{\mu}}$ of ${{\mathcal{H}_{\mu}}}$ defined by
    $\chi^{\mathcal{H}_{\mu}}(\pi_{\mu}(x_{0}))=(\Lambda,\mathcal{J},\pi_{\mu}(\mathcal{C}))$
    where $\pi_{\mu}(\mathcal{C})=\{\pi_{\mu}(c_{i}):c_{i}\in \mathcal{C}\}.$
\end{corollary}

\section{Example}\label{sec5}
Let $Q=\RR^2$, and consider $\cT Q \cong \RR^4$ endowed with the canonical symplectic form $\omega_Q = \dd q^i \wedge \dd p_i$, where $(q^i, p_i)$ are bundle coordinates induced by the canonical coordinates $(q^i)$ of $Q$. Consider the Lie group action $\Phi\colon \RR^2 \times \cT Q \to \cT Q$ of $G=\RR^2$ on $\cT Q$ given by
$$\Phi_{(a,b)} \left(q^1, q^2, p_1, p_2\right) = \left(q^1+a, q^2+a, p_1+b, p_2+b\right)\, .$$
The associated infinitesimal generators are 
$$\xi_{1}^{\cT Q} = \partial_{q^1} + \partial_{q^2}\, , \quad \xi_{2}^{\cT Q} = \partial_{p_1} + \partial_{p_2}\, .$$
Note that $\Phi$ is a symplectic action --i.e., $\Phi_{(a, b)}^\ast \omega_Q = \omega_Q$ for any $(a, b)\in \RR^2$, or, equivalently $\liedv{\xi_{1}^{\cT Q}} \omega_Q = \liedv{\xi_{2}^{\cT Q}} \omega_Q = 0$. However, it is not the cotangent lift of an action of $\RR^2$ on $Q$. Therefore, we cannot consider the natural momentum map. A momentum map $J\colon \cT Q \to \mathfrak{g}^*$ for the action $\Phi$ is given by
$$J \left(q^1, q^2, p_1, p_2 \right) = \left(p_1 + p_2, -q^1 - q^2\right)\, .$$
Here $\mathfrak{g}^\ast$ denotes the dual of the Lie algebra $\mathfrak{g} \cong \RR^2$ of $G=\RR^2$.
Since the Lie group is Abelian, the adjoint and coadjoint actions are trivial, namely,
$$\Ad_g \xi = \xi\, , \quad \Ad_g^\ast \mu = \mu\, , \quad \forall\, g\in G=\RR^2\, ,\,  \forall\, \xi\in \mathfrak{g}\, ,\, \forall\, \mu\in\mathfrak{g}^\ast\, .$$
Hence,
\begin{align*}
    J \circ \Phi_{(a, b)} \left(q^1, q^2, p_1, p_2\right)
    &- \Ad_{(a,b)^{-1}} \circ J \left(q^1, q^2, p_1, p_2\right)
    \\ &
    = J \circ \Phi_{(a, b)} \left(q^1, q^2, p_1, p_2\right)
    - J \left(q^1, q^2, p_1, p_2\right)
    \\ &
    = J \left(q^1+a, q^2+a, p_1+b, p_2+b\right)
    - J \left(q^1, q^2, p_1, p_2\right)
    \\&
    = (2b, -2a)\, ,
\end{align*}
which is independent of the point $(q^1, q^2, p_1, p_2) \in \cT Q$, and thereupon the co-adjoint cocycle is given by $\sigma(a, b) = (2b, -2a)$.

The Hamiltonian function
$$H \left(q^1, q^2, p_1, p_2\right) = \frac{(p_1-p_2)^2}{2} + V(q^{1}-q^{2})\, ,$$
where $V$ is a potential function depending only on $q^{1}-q^{2}$, is $\Phi$-invariant. Consider the hybrid Hamiltonian system $\mathcal{H}=(D,X_H,S,\Delta)$, with $X_H$ the Hamiltonian vector field of $H$, and
\begin{align*}
    & S = \left\{\left(q^1, q^2, p_1, p_2\right)\mid q^1-q^2 = c\, , \quad p_1-p_2 <0\right\}\, ,\\
    & \Delta \left(q^1, q^2, p_1, p_2\right) = \left(q^1, q^2, p_1 - \frac{1+e}{2}(p_
    1-p_2) , p_2 + \frac{1+e}{2}(p_1-p_2)\right)\, ,
\end{align*}
where $c\in \RR$ and $e\in [0, 1]$. The action $\Phi$
is an hybrid action for $\mathcal{H}$. Indeed, $\Phi$ restricts to a Lie group action of $\RR^2$ on $S$, and the impact map is equivariant with respect to this action, namely, $\Delta\circ \restr{\Phi_g}{S}=\Phi_g\circ \Delta$ for all $g\in \RR^2$. Moreover, $J$ is a hybrid momentum map, i.e., $J\circ \Delta = \restr{J}{S}$.

The isotropy subgroup with respect to the affine action $\Psi(g,\mu)=\Ad_{g^{-1}}^{*} \mu+\sigma(g)$ is given by $\tilde{G}_{\mu}=\{g=(a,b)\in \RR^2:\Ad_{g^{-1}}^{*}\mu+\sigma(a,b)=\mu\}=\{(0,0)\}$. Let $\mu=(\mu_{1},\mu_{2})\in \mathfrak{g}^{*}$ be a regular value of $J$, and consider the quotient manifold $D_{\mu}=J^{-1}(\mu)/\tilde{G}_{\mu}= J^{-1}(\mu)$, where 
\begin{align*}
J^{-1}(\mu)&=\{(q^{1},q^{2},p_{1},p_{2}):J(q^{1},q^{2},p_{1},p_{2})=\mu\}\, ,\\
&=\{(q^{1},q^{2},p_{1},p_{2}):(p_{1}+p_{2},-(q^{1}+q^{2}))=(\mu_{1},\mu_{2})\}\, .
\end{align*}
We can use $(\restr{q^2}{D_\mu}, \restr{p_2}{D_\mu})$ as coordinates in $D_\mu$. With a slight abuse of notation, we will denote them simply by $(q^2, p_2)$.
By Theorem \ref{TEO}, $\mathcal{H}_{\mu}=(D_{\mu},X_{H_{\mu}},S_{\mu},\Delta_{\mu})$ is a hybrid manifold with $X_{H_{\mu}}$ the Hamiltonian vector field of $H_{\mu}$, given by
$$H_{\mu}(q^{2},p_{2}) = \frac{(\mu_1-2p_2)^2}{2} + V(-\mu_{2}-2q^{2})\, ,$$
and
\begin{align*}
    & S_{\mu} = \left\{\left(q^2, p_2,\mu_{1},\mu_{2}\right)\mid -\mu_{2}-2q^2 = c\, , \quad \mu_{1}-2p_2 <0\right\}\, ,\\
    & \Delta_{\mu} \left(q^2, p_2,\mu_{1},\mu_{2}\right) = \left(-\mu_{2}-q^2, q^2, (\mu_{1}-p_2) - \frac{1+e}{2}(\mu_
    1-2p_2) ,
    \right. \\  & \left. \hspace{4cm} 
    p_2 + \frac{1+e}{2}(\mu_
    1-2p_2)\right)\,.
\end{align*}
\begin{credits}
\subsubsection{\ackname} L.~Colombo received financial support from Grant PID2022-137909NB-C21 funded by MICIU/AEI/10.13039/501100011033.

\subsubsection{\discintname}
The authors have no competing interests to declare.
\end{credits}
%
%
%
\renewcommand{\url}[1]{} 

\printbibliography

\end{document}